\documentclass{colt2012} % Anonymized submission

\usepackage{algorithm,algorithmic}
\usepackage{amssymb, multirow, paralist, color}
\newtheorem{thm}{Theorem}

\newtheorem{cor}[thm]{Corollary}

\def \R {\mathbb{R}}

\def \y {\mathbf{y}}
\def \x {\mathbf{x}}
\def \L {\mathcal{L}}

\def \S {\mathcal{S}}

\def \xh {\widehat{\x}}

\def \B {\mathcal{B}}

\def \R {\mathbb{R}}

\def \A {\mathcal{A}}
\def \B {\mathcal{B}}

\def \nh {\widehat{\nabla}}
\def \Sb {\overline{\S}}

\title{A New Analysis of Compressive Sensing by\\
Stochastic  Proximal Gradient Descent}
 \coltauthor{\Name{Rong Jin}$^{\dagger}$ \Email{rongjin@cse.msu.edu}\\
\Name{Tianbao Yang}$^\ddagger$ \Email{tyang@ge.com}\\
\Name{Shenghuo Zhu}$^\natural$\Email{zsh@sv.nec-labs.com}\\
   \addr$^\dagger$Department of Computer Science and Engineering\\
    Michigan State University, East Lansing, MI, 48824, USA \\
   \addr$^\ddagger$Machine Learning Lab,  GE Global Research, San Ramon, CA, 94583, USA\\
   \addr$^\natural$NEC Laboratories America, Cupertino, CA, 95014, USA
}

\begin{document}
\maketitle

\begin{abstract}
In this manuscript, we analyze the sparse signal recovery (compressive sensing) problem from the perspective of convex  optimization by stochastic proximal gradient descent. This view allows us to significantly simplify the recovery analysis of compressive sensing. More importantly, it leads to an efficient optimization algorithm for solving the regularized optimization problem related to the sparse recovery problem. Compared to the existing approaches, there are two advantages of the proposed algorithm. First, it enjoys a geometric convergence rate and therefore is computationally efficient. Second, it guarantees that the support set of any intermediate solution generated by the proposed algorithm is concentrated on the support set of the optimal solution. 

\end{abstract}

\section{Introduction and Related Work}
The problem of sparse signal recovery is to reconstruct a sparse signal given a number of linear measurements of the signal.  The problem has been studied extensively under  two closely related settings, i.e., lasso~\cite{} and compressive sensing~\cite{}. Lasso is known as a tool of model selection that aims to learn a sparse model $\beta\in\mathbb R^d$ from a data design matrix $X\in\mathbb R^{n\times d}$ and noisy measurements $\mathbf y=X\beta +\varepsilon$ of $\beta$ ,  where $\varepsilon$ are zero-mean independent Gaussian random variables, by solving the $\ell_1$ regularized least square problem $\min_{\beta\in\R^d}\|\y-X\beta\|^2_2+\lambda\|\beta\|_1$.  Compressive sensing focuses more on the study of how many random measurements are needed to optimally recover a sparse signal $\x_*\in\mathbb R^d$. In the manuscript, we provide a new perspective of compressive sensing from the viewpoint of  convex optimization by gradient descent. Our analysis reveals that in order to solve the optimal recovery problem of $\min_{\x\in\mathbb R^d}\frac{1}{2}\|\x-\x_*\|_2^2$ in hindsight by a gradient descent method, the random measurements of the signal $\x_*$ denoted by  $U\x_*$  are used for computing  a stochastic gradient of the objective. Furthermore, we develop a stochastic gradient descent method that  solves a composite gradient mapping with $\ell_1$ regularization at each iteration, which ensures the support set of intermediate solution concentrates on the support set of the optimal solution. Finally, we prove that the proposed algorithm enjoys a geometric convergence rate.  To the best of our knowledge, this work is the first that analyze the compressive sensing in the angle of optimization by stochastic gradient descent. 

A great volume of work have been devoted to the problem of sparse signal recovery in different philosophies. In the following, we briefly review some related work that solves  the optimization problem for reconstructing the optimal signal with a linear (i.e., geometric)  convergence rate.     In~\citep{bredies-2008-linear,hale-2008-fixed}, the authors established linear convergence rates as the iterates are close enough to the optimum. \citet{tropp-2007-signal} showed that if an algorithm can quickly identify the support set of the optimal solution, then the optimization is effectively reduced to a lower-dimensional subspace, and geometric convergence can be achieved.  \citet{garg-2009-gradient}  showed a geometric convergence rate for the recovered solution by a sparsification. % given  the sparsity of the target solution, making it less practical for real-world problems. 
In~\citep{agarwal-2011-fast}, the authors showed that a simple gradient descent algorithm for the constrained Lasso can achieve a global geometric convergence rate in recovering the target solution (Corollary 2)~\footnote{In the same paper, the authors also discussed a gradient descent algorithm for the regularized Lasso, which unfortunately is only able to recover the solution up to the statistical tolerance.}. One shortcoming with the analysis in~\citep{agarwal-2011-fast} is that the parameter $\kappa$ in linear convergence is lower bounded by a constant (i.e., $3/4$) independent from the number of random measurements, a disappointing feature as we expect a faster convergence with the increasing number of random measurements.

The proposed approach is similar to several existing algorithms~\citep{wen-2010-fast,wright-2009-sparse,hale-2008-fixed,xiao-2012-proximal} developed for $\ell_1$ regularized minimization in that all of them solve the regularized optimization problem by gradually shrinking the value of the regularization parameter. To the best of our knowledge, ~\citep{xiao-2012-proximal} is the only work in this direction that provides theoretical guarantee. The main difference between this work and the work~\citep{xiao-2012-proximal} is that instead of performing a simple gradient mapping for each value of the regularized parameter, the algorithm~\citep{xiao-2012-proximal} requires, at each iteration, solving the $L1$ regularized optimization problem to certain accuracy, leading to a significant computational overhead in optimization.

\section{Algorithm}

Let $\x_* \in \R^d$ be a $s$-sparse high dimensional signal to be recovered, where the number of non-zero elements in $\x_*$ is $s$. We denote by $S(\x)$ the support set for $\x$ that includes all the indices of the non-zero entries in $\x$, i.e.,
\begin{eqnarray}
    S(\x) = \left\{i\in[d]: [\x]_i \neq 0 \right\} \label{eqn:S}
\end{eqnarray}
where $[d]$ denotes the set $\{1,\ldots, d\}$ and $[\x]_i$ denote the $i$-th element in $\x$.  We also denote by $\Sb(\x) = [d]\setminus \S(\x)$ the complementary set of $\S(\x)$. In particular, we use $\S_*, \Sb_*$ to denote the support set and complementary set of $\x_*$.  Similar to most of the previous analysis, we assume that $\|\x_*\|_2 \leq R$. 

To motivate our approach, we first consider the following optimization problem
\begin{eqnarray}
    \min\limits_{\x\in\mathbb R^d} \quad \L(\x) = \frac{1}{2}\|\x - \x_*\|_2^2 \label{eqn:1}
\end{eqnarray}
Evidently, the optimal solution to (\ref{eqn:1}) is $\x_*$. We now consider a gradient descent method for optimizing the problem in (\ref{eqn:1}), leading to the following updating equation for $\x_t$
\begin{eqnarray}
    \x_{t+1} = \mathop{\arg\min}\limits_{\x\in\R^d} \left\|\x -( \x_{t} -  \nabla \L(\x_t))\right\|_2^2 \label{eqn:update-1}
\end{eqnarray}
where $\nabla \L(\x) = \x - \x_*$. Since the problem in (\ref{eqn:1}) is both smooth and strongly convex, the above updating enjoys a geometric convergence rate~\footnote{In fact, only one step is needed.}, allowing an efficient reconstruction of $\x_*$.

However, the updating rule in (\ref{eqn:update-1}) can not be used because it requires knowing $\x_*$, the full information of the sparse signal to be recovered. In compressive sensing, the only available information about the target signal $\x_*$ is its random measurements. More specifically, let $U \in \R^{m\times d}$ be a random measurement matrix and $\y=U\x_*$ be the corresponding $m$ random measurements. Using the random measurements, we construct an approximate gradient as
%Gaussian random matrix, where each entry is drawn independently from a Gaussian distribution $\N(0, 1/m)$. Using Gaussian random matrix $U$, the random measurements of $\x_*$ are given by $\y = U\x_* \in \R^m$. As a result, we can construct an unbiased estimator of the gradient as
\begin{equation}
    \nh \L(\x_t) = U^{\top}U(\x_t - \x_*) = U^{\top}(U\x_t - \y) \label{eqn:gradient}
\end{equation}
To ensure $\nh\L(\x_t)$ provide an useful estimate  of $\nabla\L(\x_t)$, we assume the random measurement matrix $U$ satisfies the following restricted isometry properties (RIP)  (with an overwhelming probability). 
\begin{definition}[$s$-restricted isometry constant]
Let  $\delta_s\geq 0$ be the smallest constant  such that for any subset $\mathcal T\in[d]$ with $|\mathcal T|\leq s$ and $\x\in\R^{|\mathcal T|}$,
\begin{align*}
(1-\delta_s)\|\x\|_2^2\leq \|U_{\mathcal T}\x\|_2^2\leq (1+\delta_s)\|\x\|_2^2
\end{align*}
where $U_{\mathcal T}$ denote the sub-matrix of $U$ with columns from $\mathcal T$. 
\end{definition}
\begin{definition}[$s,s$-restricted orthogonality  constant]
Let  $\theta_{s,s}$ be the smallest constant  such that for any two disjoint  subsets $\mathcal T, \mathcal T'\in[d]$ with $|\mathcal T|\leq s$, $|\mathcal T'|\leq s$, $2s\leq d$, and for any $\x\in\R^{|\mathcal T|}$,  $\x'\in\R^{|\mathcal T'|}$, 
\begin{align*}
|\langle U_{\mathcal T}\x,  U_{\mathcal T'}\x'\rangle| \leq \theta_{s,s}\|\x\|_2\|\x'\|_2
\end{align*}
\end{definition}
The above two constants are standard tools in the analysis of optimal recovery of compressive sensing. It has been shown that~\cite{}  several random measurement matrix including Gaussian measurement matrix, binary measurement matrix, Fourier measurement matrix and incoherent measurement matrix satisfy the above RIP  with small $\delta_s$ and $\theta_{s,s}$.  %For the example of  Gaussian random measurements $U\sim \N(0, 1/m)$, it satisfies the RIP  with a probability $1-\delta$ for $\displaystyle \delta_s\leq \sqrt{\frac{s\log(s^2+s) + \log(1/\delta)}{m}}$ and $\theta_{s,s}\leq \delta_{2s}$. 

Next, we will use $\nh\L(\x_t)$ as an approximation of $\nabla \L(\x_t)$ and update the solution by performing the following  proximal  mapping: 
%where $\E[\nh \L(\x_t)] = \nabla \L(\x_t)$ due to the fact $\E[U^{\top}U]=I$. 
%Using the unbiased estimator $\nh \L(\x_t)$, we update the solution by performing the following  composite gradient mapping
\begin{eqnarray}
    \x_{t+1} = \mathop{\arg\min}\limits_{\x \in \R^d} \tau_t \|\x\|_1 + \langle\x - \x_t, \nh \L(\x_t) \rangle + \frac{1 + \gamma}{2}\|\x - \x_t\|_2^2 \label{eqn:update-2}
\end{eqnarray}
where $\tau_t > 0$ is the regularization parameter that varies over the iterations and $\gamma > 0$ is a parameter essentially due to the RIP conditions. 
The updating rule given in (\ref{eqn:update-3}) differs from (\ref{eqn:update-1}) in that (i) the true gradient $\nabla\L(\x_t)$ is replaced with an approximate gradient $\nh\L(\x_t)$ and (ii) a  $\ell_1$ regularization term $\tau_t\|\x\|_1$ is added. With appropriate choice of $\tau_t$, this regularization term will essentially remove the noise arising from the approximate  gradient and consequentially lead to the geometric convergence rate.

\paragraph{Remark:} We note that our approach is fundamentally different from the classical idea of stochastic gradient descent. In stochastic gradient descent, we have access to the stochastic oracle of the gradients. By drawing an unbiased estimate of the gradient independently from the statistical oracle at each iteration, stochastic gradient descent is able to reduce the noise in the stochastic gradients through the average by exploring the concentration inequality of martingales. In contrast, in compressive sensing, we are only provided with {\it one} set of random measurements for the target signal  $\x_*$. Since all the estimates of gradients are based on the same set of random measurements, they are statistically dependent, making it impossible to explore the martingale technique for reducing the noise in the estimates of gradients. The $\ell_1$ regularization term in the updating rule in (\ref{eqn:update-2}) is essentially introduced to reduce the noise in the statistical gradients, and therefore plays similar role as the concentration inequality of martingales.

To give the solution of $\x_{t+1}$ in a closed form, we write~(\ref{eqn:update-2}) as 
\begin{eqnarray}
    \x_{t+1} = \mathop{\arg\min}\limits_{\x \in \R^d} \frac{1}{2}\left\|\x - \left(\x_t-\frac{1}{1+\gamma}\nh\L(\x_t)\right)\right\|_2^2 + \frac{\tau_t}{1+\gamma}\|\x\|_1 \label{eqn:update-3}
\end{eqnarray}  According to~\cite{}, the value of $\x_{t+1}$ is given by 
\begin{equation}\label{eqn:xt}
\x_{t+1} = sign(\xh_t)\left[|\xh_t| - \frac{\tau_t}{1+\gamma}\right]_+
\end{equation}
where $\xh_t = \x_t - (1/(1+\gamma))\nh\L(\x_t)$ and $[v]_+=\max(0, v)$.  
%The updating rule in~(\ref{eqn:xt}) indicates the following lemma that if the support set of $\x_t$ belongs to the support set of $\x_*$, so does the support set of $\x_{t+1}$ under some condition imposed on $\tau_t$. 
%\begin{lemma}
%If $\S(\x_{t})\subseteq \S_*$, then $\S(\x_{t+1})\subseteq\S_*$ provided 
%\[
%\tau_t\geq \max_{i\in\Sb_*} |[\nh\L(\x_t)]_i|
%\]
%\end{lemma}
%\begin{proof}
%The lemma is proved if we can show that for any $i\in\Sb_*$, $[\x_{t+1}]_i=0$. 
%Since
%\begin{align*}
%[\x_{t+1}]_i = sign([\xh_t]_i) \left[\left|[\x_t]_i -\frac{1}{1+\gamma}[\nh\L(\x_t)]_i\right| - \frac{1}{1+\gamma}\tau_t\right]_+, 
%\end{align*}
%for any $i\in \Sb_*$, $[\x_t]_i=0$ due to $\S(\x_t)\subseteq\S_*$ and therefore 
%\[
%[\x_{t+1}]_i = sign([\xh_t]_i) \frac{1}{1+\gamma}\left[\left|[\nh\L(\x_t)]_i\right| - \tau_t\right]_+
%\]
%which is zero under the condition of $\tau_t\geq \max_{i\in\Sb_*}\left|[\nh\L(\x_t)]_i\right|$. 
%\end{proof}
%Given the above analysis, 
We present the detailed steps of the proposed approach in   Algorithm~\ref{alg:1}  for reconstructing the sparse signal given a set of random measurements. 
\begin{algorithm}[t] 
\caption{A Composite Optimization Approach for Compressive Sensing}
\begin{algorithmic}[1]

\STATE {\bf Input:} Gaussian random matrix $U \in \R^{d\times m}$, random measurements $\y = U^{\top}\x_*$, regularization parameters $\tau_1, \ldots, \tau_T$, and $\gamma$

\STATE {\bf Initialize} $\x_1 = 0$.
\FOR{$t=1, \ldots, T$}
\STATE Compute $\displaystyle\xh_{t}= \x_t -\frac{1}{1+\gamma}U(U^{\top}\x_t - \y)$
    \STATE Update the solution $\displaystyle \x_{t+1} = sign(\xh_t)\left[|\xh_t| - \frac{\tau_t}{1+\gamma}\right]_+$
    \ENDFOR

\STATE {\bf Output} the final solution $\x_{T+1}$
\end{algorithmic}
\label{alg:1}
\end{algorithm}
To end this section, we present our main result in the following theomrem which states the theoretical guarantee of Algorithm~\ref{alg:1}. 
\begin{thm} \label{thm:main}
Let $\x_*\in\mathbb R^d$ be a $s$-sparse signal and $\y=U\x_*$ be a set of $m$ random measurements of $\x_*$. Set $\gamma, \tau_t$ in Algorithm~\ref{alg:1} as
\[
\gamma= \max(\delta_{3s}, \theta_{s,s} + \delta_s), \quad \tau_t =\frac{\theta_{s,s}+\delta_s + \gamma}{\sqrt{s}} (4\gamma)^{(t - 1)/2}R, t=1, \ldots, T.
\]
If we assume $\gamma\leq 1/4$, then  (i) $\|\S_t\cup\S_*\|\leq 2s$  and (ii) $\|\x_{t} - \x_*\|_2 \leq (4\gamma)^{(t-1)/2} \|\x_*\|_2$ , and (iii) $\|\x_{t} - \x_*\|_1 \leq \sqrt{s}(4\gamma)^{(t-1)/2} \|\x_*\|_2, t=1, \ldots, T$ 
\end{thm}

\section{Analysis}
Before presenting our analysis, we introduce a few notations that will be used throughout the paper. Given a set $\S \subseteq [d]$, we denote $[\x]_{\S}$ the vector that only includes the entries of $\x$ in the subset $\S$. Given two subsets $\A \subseteq [d]$ and $\B \subseteq [d]$, we denote by $[M]_{\A, \B}$ a sub-matrix that includes all the entries $(i,j)$ in matrix $M$ with $i \in \A$ and $j \in \B$. 
%We are now examining the property of the unbiased estimator of the gradient given in (\ref{eqn:gradient}). We assume that the support set for $\x_t$ (i.e. the subset of non-zero entries in $\x_t$), denoted by $\S_t$, is a subset of the support set for $\x_*$, i.e. $\S_t \subset \S_*$. As a result, the support set for $\x_t - \x_*$ is also a subset of $\S_*$. Our first result shows that under this assumption, $|[\nh \L(\x_t)]_{\S_*} - [\nabla \L(\x_t)]_{\S_*}|$ is small provided $m$ is sufficiently large.
We first prove the following Theorem. 
\begin{thm} \label{thm:sabc}
Let $\S_t$ be the support set of $\x_t$ and $\S_*$ be the support set of $\x_*$. Define $\S^c_t = \S_t\cup \S_*$, $\S^a_{t} = \S_c \setminus \S_*$.  If we assume $|\S_t\cup \S_*|\leq 2s$, at most $s$ entries of $[(1+\gamma)\x_t- U^{\top}U(\x_t-\x_*)]_{\Sb_*}$ with magnitude larger than
$
      \displaystyle   \frac{\theta_{s,s}+\theta_{s,s} + \gamma}{\sqrt{s}}\|\x_t - \x_*\|_2. 
$. 
\end{thm}
\begin{proof}
For any subset $\S' \subset \Sb_*$ of size $s$, let $\S'_1=\S'\cap\S_t^a$ and $\S'_2=\S'\setminus\S_t^a$. We have
\begin{align*}
   & \left\|[U^{\top}U(\x_t-\x_*) ]_{\S'}- (1+\gamma)[\x_t]_{\S'}\right\|_2 =  \left\|U^{\top}_{\S'}U_{\S_*}\left[\x_t - \x_*\right]_{\S_*} + U^{\top}_{\S'}U_{\S^a_t}\left[\x_t\right]_{\S^a_t} - (1+\gamma)[\x_t]_{\S'}\right\|_2 \\
    & \leq \left\|U^{\top}_{\S'}U_{\S_*}\left[\x_t - \x_*\right]_{\S_*}\right\|_2 +\left\|U_{\S'_2}^{\top}U_{\S^a_t}\left[\x_t\right]_{\S^a_t}\right\|_2  + \left\|U_{\S'_1}^{\top}U_{\S^a_t}\left[\x_t\right]_{\S^a_t} - (1+\gamma)[\x_t]_{\S'_1}\right\|_2\\
    &\leq \left\|U^{\top}_{\S'}U_{\S_*}\right\|_2\left\|\left[\x_t - \x_*\right]_{\S_*}\right\|_2 +\left\|U_{\S'_2}^{\top}U_{\S^a_t}\right\|_2\left\|\left[\x_t\right]_{\S^a_t}\right\|_2 +  \left\|U_{\S^a_t}^{\top}U_{\S^a_t}\left[\x_t\right]_{\S^a_t} - (1+\gamma)[\x_t]_{\S^a_t}\right\|_2\\
    & \leq  \theta_{s,s}\|[\x_t - \x_*]_{\S_*}\|_2 + \theta_{s,s}\|[\x_t]_{\S^a_t}\|_2 + (\delta_s+\gamma)\|[\x_t]_{\S^a_t}\|_2 \leq (\theta_{s,s}+ \delta_s + \gamma) \|\x_t - \x_*\|_2
\end{align*}
Since the above inequality holds for any subset $\S' \subseteq \Sb_*$ of size $s$, we form the set $\S'$ by including the largest $s$ entries in absolute value of $[(1+\gamma)\x_t - U^{\top}U(\x_t-\x_*) ]_{\Sb_*}$. Then the smallest absolute value in $[(1+\gamma)\x_t- U^{\top}U(\x_a-\x_b)]_{\S'}$ is bounded by $\displaystyle\frac{\theta_{s,s_a}+\theta_{s,s^c_a}}{\sqrt{s}}$. By the construction of $\S'$,  the smallest entry in $\S'$ is the $s$th largest entry in $[(1+\gamma)\x_t- U^{\top}U(\x_t-\x_*)]_{\Sb_*}$, we conclude that at most $s$ entries with magnitude larger than $ \displaystyle \frac{\theta_{s,s}+ +\delta_s+\gamma}{\sqrt{s}}\|\x_t- \x_*\|_2$.
\end{proof}
As an immediate result of Theorem~\ref{thm:sabc}, we prove the following Corollary. 
\begin{cor}\label{cor:xt}
Let $S_t$ be the support set of $\x_t$ and $\S_*$ be the support set of $\x_*$. If $|S_t\cup S_*|\leq 2s$ and $\tau_t\geq \frac{\theta_{s,s}+\delta_s+\gamma}{\sqrt{s}}\|\x_t-\x_*\|_2$, then  $|\S_{t+1}\cup \S_*|\leq 2s$ and $|\S_*\cup \S_t\cup \S_{t+1}|\leq 3s$. 
\end{cor}
\begin{proof}
As shown in~(\ref{eqn:xt}), $\x_{t+1}$ is given by 
\begin{align*}
\x_{t+1} = sign(\xh_t)\frac{1}{1+\gamma}\left[\left|(1+\gamma)\x_t - \nh\L(\x_t)\right| - \tau_t\right]_+
\end{align*}
By Theorem~\ref{thm:sabc}, we know that there are at most $s$ entries in $\left|\left[(1+\gamma)\x_t - \nh\L(\x_t)\right]_{\Sb_*}\right|$ are larger than $(\theta_{s,s}+\delta_s + \gamma)\|\x_t-\x_*\|_2/\sqrt{s}$, therefore $[\x_{t+1}]_{\Sb_*}$ has at most $s$ non-zeros entries. It concludes that $|\S_{t+1}\cup \S_*|\leq 2s$ and $|\S_*\cup \S_t\cup \S_{t+1}|\leq 3s$. 
\end{proof}

\begin{thm} \label{thm:induction}
If we assume $\|\x_t- \x_*\|_2^2\leq \Delta_t^2= (4\gamma)^{t-1}R^2$,  set  $\displaystyle\tau_t = \frac{\theta_{s,s}+\delta_s +\gamma}{\sqrt{s}}\Delta_t$ and $\gamma\geq \max(\delta_{3s}, \theta_{s,s}+\delta_s)$,  then we have
\[
    \|\x_{t+1} - \x_*\|_2^2 \leq \Delta^2_{t+1} = (4\gamma)^tR^2
\]
\end{thm}

\begin{proof}
Let $\mathcal T=\S_*\cup \S_t\cup \S_{t+1}$, by Corollary~\ref{cor:xt}, we have $|\mathcal T|\leq 3s$, therefore $\|U^{\top}_{\mathcal T}U_{\mathcal T}- I\|_2\leq \delta_{3s}$. Next, we proceed the proof as follows: 
\begin{align*}
\L(\x_{t+1})& = \L(\x_t) + \langle \x_{t+1} - \x_t, \nabla \L(\x_t) \rangle + \frac{1}{2}\|\x_{t+1} - \x_t\|_2^2 \\
& =  \L(\x_t) + \langle \x_{t+1} - \x_t, \nh \L(\x_t) \rangle + \langle \x_{t+1} - \x_t, \nabla \L(\x_t) - \nh \L(\x_t) \rangle + \frac{1}{2}\|\x_{t+1} - \x_t\|_2^2 \\
& \leq  \L(\x_t) + \langle \x_{t+1} - \x_t, \nh \L(\x_t) \rangle + \langle \x_{t+1} - \x_t, (I - U^{\top}U)(\x_t-\x_*) \rangle + \frac{1}{2}\|\x_{t+1} - \x_t\|_2^2\\
& \leq  \L(\x_t) + \langle \x_{t+1} - \x_t, \nh \L(\x_t) \rangle +   \|I - U_{\mathcal T}^{\top}U_{\mathcal T}\|_2\|\x_{t+1} - \x_t\|_2\|\x_t-\x_*\|_2 + \frac{1}{2}\|\x_{t+1} - \x_t\|_2^2\\
& \leq  \L(\x_t) + \langle \x_{t+1} - \x_t, \nh \L(\x_t) \rangle +  \delta_{3s}\|\x_{t+1} - \x_t\|_2\|\x_t-\x_*\|_2 + \frac{1}{2}\|\x_{t+1} - \x_t\|_2^2\\
& \leq  \L(\x_t) + \langle \x_{t+1} - \x_t, \nh \L(\x_t) \rangle + \frac{1 + \gamma}{2}\|\x_{t+1} - \x_t\|_2^2 + \frac{\delta_{3s}^2}{2\gamma}\|\x_t - \x_*\|_2^2\\
&\leq  \L(\x_t) + \langle \x_{t+1} - \x_t, \nh \L(\x_t) \rangle + \tau_t\|\x_{t+1}\|_1  + \frac{1 + \gamma}{2} \|\x_{t+1} - \x_t\|_2^2 \\
&+ \frac{\delta_{3s}^2}{2\gamma}\|\x_t - \x_*\|_2^2 -  \tau_t\|\x_{t+1}\|_1\\
&\leq  \L(\x_t) + \langle \x_* - \x_t, \nh \L(\x_t) \rangle + \tau_t\|\x_*\|_1  + \frac{1 + \gamma}{2} \|\x_* - \x_t\|_2^2 - \frac{1+\gamma}{2}\|\x_{t+1}-\x_*\|_2^2\\
&+ \frac{\delta_{3s}^2}{2\gamma}\|\x_t - \x_*\|_2^2 -  \tau_t\|\x_{t+1}\|_1 \quad (\text{By optimality of $\x_{t+1}$ and the strong conveity})
\end{align*}
Define
\[
\Gamma_t = \frac{\gamma + \delta_{3s}^2/\gamma}{2}\|\x_* - \x_t\|_2^2 + \tau_t\left(\|\x_*\|_1 - \|\x_{t+1}\|_1\right) + \langle \x_{*} - \x_t, \nh \L(\x_t) - \nabla \L(\x_t) \rangle - \frac{1 + \gamma}{2}\|\x_{t+1} - \x_*\|_2^2
\]
We have
\begin{align*}
\L(\x_{t+1}) & \leq   \L(\x_t) + \langle \x_* - \x_t, \nabla \L(\x_t) \rangle + \frac{1}{2}\|\x_* - \x_t\|_2^2 + \Gamma_t = \L(\x_*) + \Gamma_t = \Gamma_t
\end{align*}
where last equality follows from $\L(\x_*) = 0$. Next, we bounded $\Gamma_t$ by
\begin{align*}
 \Gamma_t  &= \frac{\gamma + \delta_{3s}^2/\gamma}{2}\|\x_* - \x_t\|_2^2 + \tau_t\left(\|\x_*\|_1 - \|\x_{t+1}\|_1\right) + \langle \x_{*} - \x_t, \nh \L(\x_t) - \nabla \L(\x_t) \rangle \\
 &\hspace*{0.2in}-  \frac{1+\gamma}{2}\|\x_{t+1} - \x_*\|_2^2 \\
& \leq  \frac{\gamma + \delta_{3s}^2/\gamma}{2}\Delta_t^2 + \frac{\theta_{s,s}+\delta_s+\gamma}{\sqrt{s}}\Delta_t\sqrt{s}\|\x_{t+1} - \x_*\|_2  + \langle \x_{*} - \x_t, (I-U^{\top}U)(\x_*-\x_t)\rangle \\
 &\hspace*{0.2in}-  \frac{1+\gamma}{2}\|\x_{t+1} - \x_*\|_2^2 \\
& \leq  \left(\frac{\gamma + \delta_{3s}^2/\gamma}{2} + \frac{(\theta_{s,s}+\delta_s+\gamma)^2}{2}\right)\Delta_t^2 + \frac{1}{2}\|\x_{t+1} - \x_*\|_2^2  + \delta_{2s}\|\x_t-\x_*\|^2_{2}-  \frac{1+\gamma}{2}\|\x_{t+1} - \x_*\|_2^2 \\
& \leq  \left(\frac{\gamma + \delta_{3s}^2/\gamma}{2} +\frac{(\theta_{s,s}+\delta_s+\gamma)^2}{2} +\delta_{2s}\right)\Delta_t^2 -  \frac{\gamma}{2}\|\x_{t+1} - \x_*\|_2^2 \\
\end{align*}
Since $\L(\x_{t+1})=\|\x_{t+1}-\x_*\|_2^2/2$,  we have
\[
\frac{1 + \gamma}{2}\|\x_{t+1} - \x_*\|^2 \leq \left(\frac{\gamma + \delta_{3s}^2/\gamma}{2} + \frac{(\theta_{s,s}+\delta_s+\gamma)^2}{2} + \delta_{2s} \right)\Delta_t^2
\]
leading to
\[
\|\x_{t+1} - \x_*\|_2^2 \leq \frac{1}{1 + \gamma}\left(\gamma + \frac{\delta_{3s}^2}{\gamma} + 2\delta_{2s} + (\theta_{s,s} + \delta_s + \gamma)^2\right)\Delta_t^2
\]
Since $\delta_s$ is no-decreasing in $s$, if we assume $\gamma\geq \max(\delta_{3s}, \theta_{s,s} + \delta_s)$, we have
\[
\|\x_{t+1} - \x_*\|_2^2 \leq \frac{4\gamma + 4\gamma^2}{1 + \gamma}\Delta_t^2 \leq 4\gamma\Delta_t^2 = \Delta^2_{t+1}
\]
\end{proof}

%\begin{proof}[Proof of Theorem~\ref{thm:main}]
%
%Theorem~\ref{thm:main} follows immediately Theorem~\ref{thm:induction}.
%\end{proof}

%\bibliographystyle{abbrv}
\bibliography{cs-opt}
\end{document}